\documentclass[12pt]{article}

\usepackage{amsmath, amsthm, amsfonts, amssymb, graphicx, tikz}
\usetikzlibrary{arrows}
\newtheorem{theorem}{Theorem}
\newtheorem{lemma}[theorem]{Lemma}

\newtheorem{corollary}[theorem]{Corollary}

\title{Length $3$ Edge-Disjoint Paths \\and Partial Orientation}

\author{Hannah Alpert \\
\small Department of Mathematics \\[-0.8ex]
\small Massachusetts Institute of Technology \\[-0.8ex]
\small Building 2, Room 229 \\[-0.8ex]
\small 77 Massachusetts Avenue \\[-0.8ex]
\small Cambridge, MA 02139 \\[-0.8ex]
\small \texttt{hcalpert@math.mit.edu} \and
Jennifer Iglesias \\
\small Department of Mathematics \\[-0.8ex]
\small Harvey Mudd College \\[-0.8ex]
\small 340 E. Foothill Blvd. \\[-0.8ex]
\small Claremont, CA 91711 \\[-0.8ex]
\small \texttt{jiglesias@hmc.edu}}

\begin{document}

\maketitle



\begin{abstract}
In 2003, it was claimed that the following problem was solvable in polynomial time: do there exist $k$ edge-disjoint paths of length exactly $3$ between vertices $s$ and $t$ in a given graph?  The proof was flawed, and we show that this problem is NP-hard even if we disallow multiple edges.  We use a reduction from \textsc{Partial Orientation}, a problem recently shown by P\'alv\"olgyi to be NP-hard.
\end{abstract}

In \cite{bley03}, Bley discussed the problem \textsc{Max Edge-Disjoint Exact-$\ell$-Length Paths}, abbreviated $\mathsf{MEDEP}(\ell)$: given an undirected multigraph, and two vertices $s$ and $t$, do there exist $k$ edge-disjoint paths between $s$ and $t$ of length exactly $\ell$?  He used a reduction to network flow to claim that $\mathsf{MEDEP}(3)$ was solvable in polynomial time, but the reduction was flawed, as we describe in~\cite{alpert12pub}, where we also showed that $\mathsf{MEDEP}(3)$ is NP-hard.  For completeness we duplicate the proof below, as Theorem~\ref{medep}.

In this note, we show that the problem $\mathsf{MEDEP}(3)$ remains NP-hard even if we require the input graph to be simple.  We call this restricted problem \textsc{Simple Max Edge-Disjoint Exact-3-Length Paths}, or $\mathsf{SMEDEP}(3)$.  To show that $\mathsf{MEDEP}(3)$ and $\mathsf{SMEDEP}(3)$ are NP-hard, we use a polynomial reduction from the problem \textsc{Partial Orientation}, shown to be NP-hard by P\'alv\"olgyi in \cite{palvolgyi09}: given a graph, can we replace some of the edges by directed edges, such that each vertex has prescribed in-, out-, and undirected-degree?

\begin{theorem}\label{medep}
The problem $\mathsf{MEDEP}(3)$ is NP-hard.
\end{theorem}

\begin{proof}
Let $G$ be the given graph in an instance of \textsc{Partial Orientation}. We construct a graph $G'$, the input to $\mathsf{MEDEP}(3)$, as follows.  To $G$ we add two new vertices $s$ and $t$.  Let $s$ be adjacent to each vertex $v$ of $G$ with multiplicity equal to the prescribed out-degree of $v$, and let $t$ be adjacent to $v$ with multiplicity equal to the prescribed in-degree of $v$.  Then the sum of the prescribed out-degrees is the degree of $s$, and the sum of the prescribed in-degrees is the degree of $t$.  If these sums are not the same, then trivially the instance of \textsc{Partial Orientation} has no solution, so we assume the degrees of $s$ and $t$ are equal.  As input to $\mathsf{MEDEP}(3)$, we set $k$ equal to the degree of $s$ and $t$.

Now, any solution of \textsc{Partial Orientation} on $G$ corresponds to a solution of $\mathsf{MEDEP}(3)$ on $G'$, and vice versa.  We simply correspond each directed edge $(u, v)$ in \textsc{Partial Orientation} with a path $suvt$ in $\mathsf{MEDEP}(3)$.  This is a polynomial reduction from \textsc{Partial Orientation} to $\mathsf{MEDEP}(3)$, and thus shows that $\mathsf{MEDEP}(3)$ is NP-hard.
\end{proof}

Note that this proof requires allowing duplicate edges incident to $s$ and to $t$.  Below we show that the problem $\mathsf{SMEDEP}(3)$, which disallows duplicate edges, is still NP-hard.

\begin{theorem}\label{deg1}
The problem \textrm{\textsc{Partial Orientation}} is still NP-hard even if the prescribed in-degree and prescribed out-degree of each vertex must be at most $1$.
\end{theorem}

The reduction presented in the proof of Theorem~\ref{medep} immediately implies the following corollary:

\begin{corollary}\label{smedep}
The problem $\mathsf{SMEDEP}(3)$ is NP-hard.
\end{corollary}

To prove Theorem~\ref{deg1}, we modify P\'alv\"olgyi's reduction from \textsc{3-SAT} to avoid using vertices of prescribed in- or out-degree greater than $1$.  For each vertex $v$ in the graph $G$, he denotes the prescribed in-degree by $\rho(v)$, the prescribed out-degree by $\delta(v)$, and the remaining prescribed undirected-degree by $\theta(v)$.  He labels each vertex with a string of letters $\rho$, $\delta$, and $\theta$, each with multiplicity equal to the value of that function at that vertex.

We will construct $G$ in parts, and within each part, we will assign the prescribed degrees to each vertex before having constructed all the edges incident to that vertex.  Thus for convenience, by a ``subgraph'' of $G$ we mean the induced subgraph on a collection of vertices, together with the prescribed in-, out-, and undirected-degrees in $G$ of those vertices, and any edges incident to those vertices with unspecified other endpoints in $G$.  We temporarily imagine each of these edges to have only one endpoint, and call them ``unfinished edges'' with respect to the subgraph.  An ``orientation'' of a graph or subgraph chooses each edge (including any unfinished edges) to be oriented in one direction or the other, or to be undirected, in a way that satisfies the prescribed in-, out-, and undirected-degrees of the vertices.

\begin{lemma}[P\'alv\"olgyi~\cite{palvolgyi09}]
For any $m \in \mathbb{N}$, there exists a tree subgraph with the following properties:
\begin{itemize}
\item  Every vertex has degree type $\rho\delta$, $\rho\theta$, $\delta\theta$, or $\rho\delta\theta$.
\item The subgraph has only two possible orientations, called ``true'' and ``false''.
\item Among the unfinished edges, there are $m$ edges $t_1, \ldots, t_m$ that point (down) away from the tree in the ``true'' orientation and (up) toward the tree in the ``false'' orientation, and $m$ unfinished edges $f_1, \ldots, f_m$ that point (up) toward the tree in the ``true'' orientation and (down) away from the tree in the ``false'' orientation.  
\end{itemize}
The size of this subgraph is linear in $m$.
\end{lemma}

P\'alv\"olgyi uses one tree subgraph for each variable $x_i$ of the given instance of \textsc{3-SAT}.  He chooses $m$ larger than both the number of times $x_i$ appears in any clause and the number of times $\overline{x_i}$ appears in any clause.  Thus he can reserve one unfinished edge $t_j$ for each time $x_i$ appears in a clause and one unfinished edge $f_k$ for each time $\overline{x_i}$ appears in a clause.  Each clause is represented by a subgraph containing the other end of each of three edges $t_j$ or $f_k$, corresponding to the three variables $x_i$ or $\overline{x_i}$ that appear in that clause.  P\'alv\"olgyi's subgraph for the clause is the only place in the construction where the prescribed in- and out-degrees are ever greater than $1$, so we propose a substitute subgraph, as follows.

\begin{lemma}
There exists a subgraph including three unfinished edges $e_1$, $e_2$, and $e_3$ such that the assignment of at least one of these three to point (down) toward the subgraph and the others (up) away from the subgraph can be extended to an orientation of the whole subgraph, but the assignment of all three to point (up) away from the subgraph cannot.  Furthermore, the subgraph can be constructed so that $\rho(v) \leq 1$ and $\delta(v) \leq 1$ for every vertex $v$ of the subgraph.
\end{lemma}

\begin{proof}
We exhibit this subgraph in Figure~\ref{our-gadget}, with $e_1$, $e_2$, and $e_3$ extending up above the subgraph.
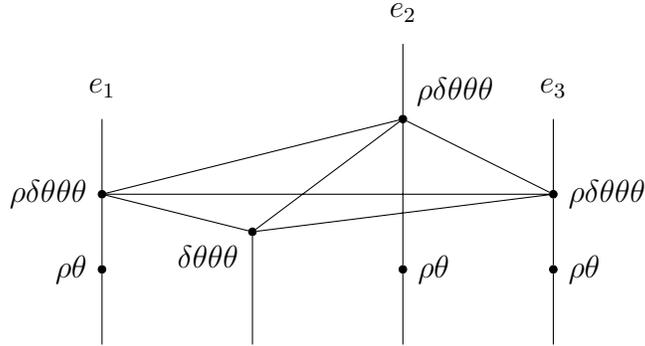
\begin{figure}[h!]
\begin{center}
\begin{tikzpicture}[ vert/.style={circle, draw=black, fill=black, inner sep = 0pt, minimum size = 1mm}]
\node(e1) at (0, 1) [label = above: $e_1$]{};
\node(e2) at (4, 2) [label = above: $e_2$]{};
\node(e3) at (6, 1) [label = above: $e_3$]{};
\node(va) at (0, 0) [vert] [label = left: $\rho\delta\theta\theta\theta$]{};
\node(vb) at (4, 1) [vert] [label = above right: $\rho\delta\theta\theta\theta$]{};
\node(vc) at (6, 0) [vert] [label = right: $\rho\delta\theta\theta\theta$]{};
\node(vaa) at (0, -1) [vert] [label = left: $\rho\theta$]{};
\node(vbb) at (4, -1) [vert] [label = right: $\rho\theta$]{};
\node(vcc) at (6, -1) [vert] [label = right: $\rho\theta$]{};
\node(vx) at (2, -.5) [vert] [label = below left: $\delta\theta\theta\theta$]{};

\draw (va)--(vb)--(vc)--(va)
(vx)--(va) (vx)--(vb) (vx)--(vc)
(0, 1)--(va)--(vaa)--(0, -2) (4, 2)--(vb)--(vbb)--(4,-2) (6, 1)--(vc)--(vcc)--(6, -2)
(vx)--(2, -2);

\end{tikzpicture}
\end{center}
\caption{Our subgraph for representing a clause.}\label{our-gadget}
\end{figure}

Figure~\ref{3-our-gadget} shows the orientations in which at least one of $e_1$, $e_2$, and $e_3$ points down and the others point up.
 
\begin{figure}[h!]
\begin{center}
\begin{tikzpicture}[ scale=1, >=triangle 45, vert/.style={circle, draw=black, fill=black, inner sep = 0pt, minimum size = 1mm}]
\node(va) at (0, 0) [vert] {};
\node(vb) at (2, 1) [vert] {};
\node(vc) at (3, 0) [vert] {};
\node(vaa) at (0, -1) [vert] {};
\node(vbb) at (2, -1) [vert] {};
\node(vcc) at (3, -1) [vert] {};
\node(vx) at (1, -.5) [vert] {};

\draw[dotted] (va)--(vb);
\draw[dotted] (vb)--(vc);
\draw[dotted] (vc)--(va);

\draw[dotted] (vx)--(va);
\draw[dotted] (vx)--(vb);
\draw[dotted] (vx)--(vc);

\draw[->] (0, 1)--(va);
\draw[->] (va)--(vaa);
\draw[dotted] (vaa)--(0, -2);

\draw[->] (2, 2)--(vb);
\draw[->] (vb)--(vbb);
\draw[dotted] (vbb)--(2,-2);

\draw[->] (3, 1)--(vc);
\draw[->] (vc)--(vcc);
\draw[dotted] (vcc)--(3, -2);

\draw[->] (vx)--(1, -2);

\end{tikzpicture}\hspace{20pt}
\begin{tikzpicture}[ scale=1, >=triangle 45, vert/.style={circle, draw=black, fill=black, inner sep = 0pt, minimum size = 1mm}]
\node(va) at (0, 0) [vert] {};
\node(vb) at (2, 1) [vert] {};
\node(vc) at (3, 0) [vert] {};
\node(vaa) at (0, -1) [vert] {};
\node(vbb) at (2, -1) [vert] {};
\node(vcc) at (3, -1) [vert] {};
\node(vx) at (1, -.5) [vert] {};

\draw[dotted] (va)--(vb);
\draw[dotted] (vb)--(vc);
\draw[dotted] (vc)--(va);

\draw[dotted] (vx)--(va);
\draw[dotted] (vx)--(vb);
\draw[->] (vx)--(vc);

\draw[->] (0, 1)--(va);
\draw[->] (va)--(vaa);
\draw[dotted] (vaa)--(0, -2);

\draw[->] (2, 2)--(vb);
\draw[->] (vb)--(vbb);
\draw[dotted] (vbb)--(2,-2);

\draw[<-] (3, 1)--(vc);
\draw[dotted] (vc)--(vcc);
\draw[<-] (vcc)--(3, -2);

\draw[dotted] (vx)--(1, -2);

\end{tikzpicture}\hspace{20pt}
\begin{tikzpicture}[ scale=1, >=triangle 45, vert/.style={circle, draw=black, fill=black, inner sep = 0pt, minimum size = 1mm}]
\node(va) at (0, 0) [vert] {};
\node(vb) at (2, 1) [vert] {};
\node(vc) at (3, 0) [vert] {};
\node(vaa) at (0, -1) [vert] {};
\node(vbb) at (2, -1) [vert] {};
\node(vcc) at (3, -1) [vert] {};
\node(vx) at (1, -.5) [vert] {};

\draw[->] (va)--(vb);
\draw[dotted] (vb)--(vc);
\draw[dotted] (vc)--(va);

\draw[dotted] (vx)--(va);
\draw[dotted] (vx)--(vb);
\draw[->] (vx)--(vc);

\draw[->] (0, 1)--(va);
\draw[dotted] (va)--(vaa);
\draw[<-] (vaa)--(0, -2);

\draw[<-] (2, 2)--(vb);
\draw[dotted] (vb)--(vbb);
\draw[<-] (vbb)--(2,-2);

\draw[<-] (3, 1)--(vc);
\draw[dotted] (vc)--(vcc);
\draw[<-] (vcc)--(3, -2);

\draw[dotted] (vx)--(1, -2);

\end{tikzpicture}
\end{center}
\caption{Orientations with at least one of $e_1$, $e_2$, or $e_3$ pointing down.}\label{3-our-gadget}
\end{figure}
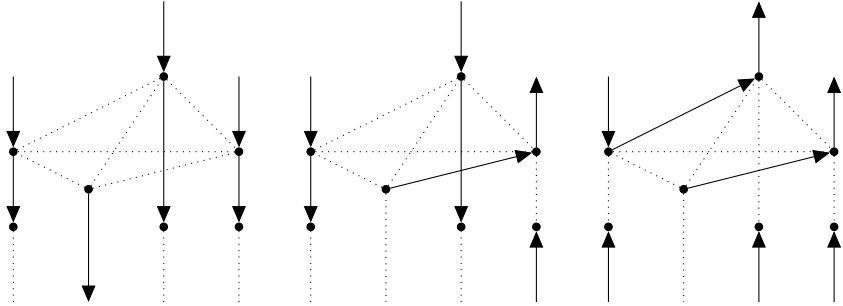
\end{proof}

Note that $e_1$, $e_2$, and $e_3$ must point either up or down because of their status as $t_j$ or $f_k$ edges in the tree subgraphs.  Thus we may ignore any orientations of the clause subgraph in which $e_1$, $e_2$, or $e_3$ is undirected, because these orientations are not possible when the subgraph is connected to the whole graph.

At this point in the construction, we have a tree subgraph for each variable, and a clause subgraph for each clause, with the subgraphs connected according to which variables appear in each clause.  In order to resolve any remaining unfinished edges in the graph, we complete the construction as P\'alv\"olgyi does: he takes the graph constructed so far and adds the \emph{mirrored reflection} of it.  That is, start by duplicating the construction.  Then, for each new vertex $v'$ corresponding to existing vertex $v$, put $\rho(v') = \delta(v)$, $\delta(v') = \rho(v)$, and $\theta(v') = \theta(v)$.  Finally, add an edge between $v$ and $v'$ if and only if $v$ was incident to an unfinished edge.  This completes the construction.  As in P\'alv\"olgyi's proof, the constructed instance of \textsc{Partial Orientation} has a solution if and only if the given instance of \textsc{3-SAT} has a solution, and our modified construction does not use any vertices of prescribed in-degree or out-degree greater than $1$.

\section*{Acknowledgments}

This research was supervised by Garth Isaak at the Lafayette College REU, which was supported by Lafayette College and the National Science Foundation (grant number DMS 0552825).  Thanks also to Brad Alpert for helping to edit the writing.

\bibliographystyle{plain}
\bibliography{edgedisjointlong}
\end{document}